\documentclass[12pt]{amsart}
\usepackage[foot]{amsaddr}
\usepackage{amsmath,amssymb,latexsym} 
\usepackage{fullpage}
\usepackage{etoolbox}
\usepackage{enumitem}
\usepackage{color,xcolor}
\usepackage[colorlinks,linkcolor=blue,citecolor=magenta,urlcolor=blue]{hyperref}
\usepackage[colorinlistoftodos,bordercolor=orange,backgroundcolor=orange!20,linecolor=orange,textsize=scriptsize]{todonotes}
\usepackage{soul}
\usepackage{microtype}
\usepackage{yhmath}
\usepackage{frenchineq}
\usepackage{prooftree}
\usepackage{cleveref}

\usepackage{tikz}
\usetikzlibrary{matrix,decorations.pathreplacing,calc,backgrounds,patterns,positioning}

\newtheorem{theorem}{Theorem} 

\newtheorem{proposition}{Proposition}
\newtheorem{property}{Property}

\newtheorem{lemma}{Lemma}

\theoremstyle{definition}
\newtheorem*{remark}{Remark}

\newtheorem{example}{Example}

\def\ds{\displaystyle}

\DeclareMathOperator{\malt}{alt}

\newcommand{\tif}{\text{if }}
\DeclareMathOperator{\unb}{unb}
\newcommand{\Bool}{\mathbb{B}}

\newcommand{\mO}{\mathbb{O}}

\def\sign{\operatorname{sign}}

\makeatletter
\newcommand{\leqnomode}{\tagsleft@true\let\veqno\@@leqno}
\newcommand{\reqnomode}{\tagsleft@false\let\veqno\@@eqno}
\makeatother

\title{$O_n$ is an $n$-MCFL}

\author{Kilian Gebhardt\(^1\)}
\address{\(^1\)Technische Universit\"at Dresden, Germany, \normalfont{\href{mailto:kilian@gebhardt.xyz}{kilian@gebhardt.xyz}}.}
\author{Fr\'ed\'eric Meunier\(^2\)}
\address{\(^2\)CERMICS, \'Ecole des Ponts ParisTech, France, \normalfont{\href{mailto:frederic.meunier@enpc.fr}{frederic.meunier@enpc.fr}}.} 
\author{Sylvain Salvati\(^3\)}
\address{\(^3\)Universit\'e de Lille, INRIA, CRIStAL, France, \normalfont{\href{mailto:sylvain.salvati@univ-lille.fr}{sylvain.salvati@univ-lille.fr}}.}

\makeatletter
\hypersetup{pdftitle={O\_n is an n-MCFL},pdfauthor={Kilian Gebhardt, Frédéric Meunier, and Sylvain Salvati}}
\makeatother

\begin{document}

\begin{abstract}
  Commutative properties in formal languages pose problems at the frontier of
  computer science, computational linguistics and computational group theory. A
  prominent problem of this kind is the position of the language $O_n$, the
  language that contains the same number of letters $a_i$ and $\bar a_i$ with
  $1\leq i\leq n$, in the known classes of formal languages. It has recently
  been shown that $O_n$ is a Multiple Context-Free Language (MCFL). However the
  more precise conjecture of Nederhof that $O_n$ is an MCFL of dimension $n$ was
  left open. We present two proofs of this conjecture, both relying
  on tools from algebraic topology. 
  On our way, we prove a variant of the necklace splitting theorem.
\end{abstract}

\maketitle

\section{Introduction}

The language $O_n$ is the language built on the alphabet $\Sigma_n = \{a_i,\bar
a_i\mid i \in [n]\}$ and that contains exactly all words $w$ which, for
every $i$ in $[n]$, have the same number of occurrences of $a_i$ and $\bar
a_i$. Writing $|w|_{c}$ the number of occurrences of the letter $c$ in $w$, this condition becomes $|w|_{a_i} =
|w|_{\bar a_i}$ for all $i$ in $[n]$. The problem of situating the languages $O_n$ within known
classes of languages has been raised at least in two different communities.

The first one is that of computational linguistics. This problem has attracted
attention with a language called $MIX$ (also called the Bach language as it was
introduced by 
Bach~\cite{bach81:_discontinuous_constituents_generalized_categorioal_grammars,
  bach88:_categorial_grammar_as_theories_of_languages,
  pullum83:_context_freeness_and_the_computer_processing_of_human_languages})
that is rationally equivalent to $O_2$. It was related to the research program
of
Joshi~\cite{joshi85:_tree_adjoining_grammars:_how_much_context_sensitivity_is_required_to_provide_reasonnable_structural_descriptions}
which consists in formally describing the class of languages corresponding to
Natural Languages. According to Joshi's terminology, this is the class of Mildly
Context Sensitive Languages. This program tries to give abstract properties of
this class while describing possible candidates
\cite{weir_phd,joshi91:_the_convergence_of_mildly_context_sensitive_grammar_formalisms}.
Among such candidates, a powerful one is formed by Multiple Context Free
Languages (MCFL)~\cite{seki91:_multiple_context_free_grammars}. According to
Joshi \textit{et
  al.}~\cite{joshi91:_the_convergence_of_mildly_context_sensitive_grammar_formalisms}
the language ``$MIX$ can be regarded as the extreme case of free word order''
 and Mildly Context Sensitive Languages should ``perhaps'' not contain $MIX$ (and thus
$O_2$). This paper also mentions that the position of $MIX$ in classes of
languages, such as Tree Adjoining Languages, is not known and difficult to establish. In particular, it stresses
that it is not known whether $MIX$ is an
Indexed Language.

The second community that has also shown interest for the problem is that of computational
group theory which tries to identify properties of groups by means of
properties of their word problem. The word problem for a group consists in
describing the language of words that are equal to zero for a given presentation
(all presentations giving rise to rationally equivalent languages). For example,
Muller and Schuppe have characterized virtually free groups as exactly those groups
whose word problems are solved by context free
grammars~\cite{muller83:_group_theory_the_theory_of_ends_and_context_free_languages}.
A question that has been raised by that community is whether $O_2$---which
coincides with the
language corresponding to the word problem for the additive group
$(\mathbb{Z}^2,{+})$---is an Indexed
Language~\cite{gilman05:_formal_languag_applic_combin_group_theor}. This
question remains open. MCFL form a natural generalization of ``copyless''
Macro-Languages (see~\cite{DBLP:journals/ieicet/SekiK08} for a discussion) and
Macro-Languages are equivalent to Indexed Languages. This makes the
question of whether $O_2$ belongs to MCFL relevant.

A first important result in that line of research is that $O_2$ is not a
well-nested MCFL of dimension $2$~\cite{kanazawa_salvati12:_mix}, which solves a
long-standing open problem raised by
Joshi~\cite{joshi85:_tree_adjoining_grammars:_how_much_context_sensitivity_is_required_to_provide_reasonnable_structural_descriptions}.
Subsequently it has been shown that it is actually an
MCFL~\cite{salvati15:mix_2mcfl} and more precisely an MCFL of dimension $2$ (a
$2$-MCFL). Nederhof~\cite{DBLP:conf/acl/Nederhof16} has given a similar proof of
the same result. Later he conjectured that $O_n$ is an $n$-MCFL~\cite{Nederhof-2017-free-word-orders-and-MCFLs}.
As pointed in this work, a simple pumping argument shows
that $O_n$ cannot be an MCFL of dimension strictly smaller than $n$. More recently,
a breakthrough has been achieved by 
Ho~\cite{Thewordproblemofnisamultiplecontextfreelanguage} who proved that $O_n$
is an MCFL for every $n$. However the MCFL constructed in the proof
is of dimension larger than $n$, namely $8\left\lfloor \frac{n+1}{2} \right\rfloor-2$. All 
proofs related to these results are based on algebraic topology. While the
proofs of \cite{salvati15:mix_2mcfl} and
\cite{DBLP:conf/acl/Nederhof16} strongly rely on topological properties of the plane 
(existence of winding numbers of curves around points), the aforementioned proof
by Ho is based on the well-known Borsuk--Ulam theorem, a powerful theorem from algebraic topology which holds in any dimension. More
specifically, 
it uses a combinatorial application of this theorem, due to Alon 
and West~\cite{alon1986borsuk}: 
the necklace splitting theorem. 
In this paper, we also rely on related tools to prove Nederhof's conjecture,
namely that $O_n$ is an $n$-MCFL. Nederhof actually conjectures that a
particular Multiple Context-Free Grammar (MCFG) of dimension $n$ defines $O_n$. 
We prove a slightly stronger result by
showing that a grammar that uses a more restricted kind of rules is sufficient
to define $O_n$.

The article is structured as follows: In \Cref{sec:prelim}, we introduce notation regarding formal 
languages and MCFG as well as the grammar \(G_n\), which is an MCFG of dimension $n$.
\Cref{sec:main-result} establishes the main result namely that the language
of \(G_n\) is \(O_n\) using a decomposition lemma.
The decomposition lemma can on the one hand be derived from a variant of 
the necklace splitting theorem that we prove in \Cref{sec:necklace}. 
Alternatively, it can be obtained via purely combinatorial
proofs presented in \Cref{sec:combinatorial-proofs}.

\section{Background on Multiple Context Free Grammars}
\label{sec:prelim}

We write $[n]$ for the set $\{1,\dots,n\}$. For a given finite set $\Sigma$,
also called \emph{alphabet}, we write $\Sigma^\ast$ for the monoid freely
generated by $\Sigma$, and $\Sigma^{+}$ for the free semigroup generated by
$\Sigma$. The elements of $\Sigma$ are called \emph{letters} while the elements
of $\Sigma^\ast$ and $\Sigma^{+}$ are called \emph{strings} or \emph{words} and
we write $\varepsilon$ for the empty word. Given a word $w$, we write $|w|$ for
its length, and $|w|_c$ for the number of occurrences of the letter $c$ in
$w$. A \emph{language} is a subset of $\Sigma^\ast$.

We define the language $O_n$ as $\{w\in \Sigma_n^\ast\mid |w|_{a_i} =
|w|_{\bar a_i}\text { for } i \in [n]\}$ where $\Sigma_n$ is the alphabet $\{a_i,\bar a_i \mid i \in [n]\}$. 
For $\alpha$ in $\Sigma_n$, writing
$\bar\alpha$ exchanges $a_i$ and $\bar a_i$: if $\alpha$ is $a_i$, then $\bar \alpha$ is $\bar a_i$; if $\alpha$ is $\bar a_i$, then $\bar \alpha$ is $a_i$.
Two letters $\alpha$ and $\beta$ of $\Sigma_n$ are \emph{compatible}
when $\alpha = \bar\beta$. We extend the $(\bar \cdot)$ operation to words in
$\Sigma_n^\ast$ as follows: the word $\bar w$ is obtained from $w$ by applying $(\bar\cdot)$ to
each of its letters.

A ranked alphabet $\Omega$ is a pair $(\mathcal{A},\rho)$ where $\mathcal{A}$ is
a finite set and $\rho$ is a function from $\mathcal{A}$ to $\mathbb{N}$. For
$a$ in $\mathcal{A}$, the integer $\rho(a)$ is the \emph{rank} of $a$. We shall write
$\Omega^{(n)}$ for the set $\{a \in \mathcal{A} \mid \rho(a)=n\}$. The
\emph{dimension} of a ranked alphabet is the maximal rank of its elements.

A Multiple Context Free Grammar (MCFG) $G$ is a tuple $(\Omega, \Sigma, R, S)$
where $\Omega$ is a ranked alphabet of \emph{non-terminals}, $\Sigma$ is a
finite set of \emph{letters}, $R$ is a set of \emph{rules} and $S$ is an element
of $\Omega^{(1)}$ called \emph{initial non-terminal}. The rules in $R$ are of the form
\begin{equation}\label{eq:rule}
A(w_1,\ldots, w_n) \Rightarrow B_1 (x_{1,1},\ldots, x_{1,l_1}), \ldots,
B_p(x_{p,1},\ldots, x_{p,l_p})
\end{equation}
where $A$ is in $\Omega^{(n)}$, $B_k$ is in
$\Omega^{(l_k)}$, the $x_{k,j}$ are pairwise distinct variables and the
$w_j$ are elements of $(\Sigma\cup X)^\ast$ with $X = \{x_{k,j} \mid k \in
[p] \land j\in [l_k]\}$ and with the restriction that each $x_{k,j}$ may have at most
one occurrence in the string $w_1\cdots w_n$. Note that $p$ may be
equal to $0$, in which case the right part of the rule (the one on the right of
the $\Rightarrow$ symbol) is empty. Then we may write the rule by
omitting the symbol $\Rightarrow$. The \emph{dimension} of an MCFG is that of its
ranked alphabet of non-terminals. An MCFG of dimension at most $n$ is an
$n$-MCFG.

An MCFG such as $G$ defines \emph{judgments} of the form $\vdash_G A(s_1,\ldots,
s_n)$ where $A$ is in $\Omega^{(n)}$ and the $s_j$ belong to $\Sigma^\ast$. 
The notion of derivable judgments is defined inductively:
suppose we are given $p$ derivable judgments $\vdash_{G} B_k(s_{k,1},\ldots,
s_{k,l_k})$ where $B_k \in \Omega^{(l_k)}$ for $k$ in $[p]$.
For each rule of the form~\eqref{eq:rule}, the judgment 
$\vdash_G A(s_1, \ldots, s_n)$ is \emph{derivable} 
when each $s_j$ is obtained from $w_j$ by replacing each occurrence
of the variable $x_{k,j}$ by $s_{k,j}$.
The language defined by $G$, denoted by $\mathcal{L}(G)$, is the set
$\{w \in \Sigma^\ast \mid \vdash_G S(w) \mbox{ is derivable}\}$. The class of
languages that are definable by MCFGs is the class of \emph{Multiple
Context-Free Languages} (MCFL). Likewise, the class of languages definable by
$n$-MCFGs is the class of \emph{$n$-Multiple Context-Free Languages} ($n$-MCFL).

We define now $G_n$, the central $n$-MCFG for which we prove that 
it generates $O_n$. 
It uses two non-terminals $S$ and $I$ that are
respectively of rank $1$ and $n$. 
The non-terminal $S$ is the initial one.
The alphabet of $G_n$ is $\Sigma_n$.
The rules of the grammar $G_n$ are the following:
\begin{enumerate}
	\item\label{rule:init} $S(x_1\cdots x_n)\Rightarrow I(x_1,\dots,x_n)$.
	\item\label{rule:binary} $I(w_1,\dots, w_n)\Rightarrow I(x_1,\dots,x_n), I(y_1,\dots,y_n)$
	
	\qquad for all \(w_1, \ldots, w_n \in \{x_1, \ldots, x_n, y_1, \ldots, y_n\}^\ast\) such that 
	$w_1\cdots w_n = x_1y_1\cdots x_ny_n$.
	\item\label{rule:const} $I(w_1,\dots,w_n)\Rightarrow I(x_1,\dots,x_n)$ for all \(w_1, \ldots, w_n\) and all $\alpha\in\Sigma_n$, $k,\ell \in [n]$ such that $w_j = x_j$ for $j \neq k,\ell$ and
	\begin{itemize}
		\item $k \neq \ell$ implies $w_k \in \{\alpha x_k, x_k \alpha\}$ and  $w_{\ell} \in \{\bar\alpha x_\ell, x_\ell \bar\alpha\}$.
		\item $k = \ell$ implies $w_k =\alpha x_k \bar\alpha$.
	\end{itemize}
	\item\label{rule:empty} $I(\varepsilon,\dots,\varepsilon)$.
\end{enumerate}

Items numbered \eqref{rule:binary} and \eqref{rule:const} describe finite sets
of rules. The rules~(\ref{rule:binary}) are parametrized by a particular factorization
$(w_1,\dots,w_n)$ of $x_1y_1\cdots x_ny_n$. For example, when $n=3$ letting $w_1
= x_1y_1x_2$, $w_2 = y_2x_3$ and $w_3 =y_3$ is such a factorization; we have
$w_1w_2w_3 = x_1y_1x_2y_2x_3y_3$. A rule of the form~(\ref{rule:const}) adds a compatible pair of letters at distinct endpoints of the words. For
example, $I(a_1x_1,x_2,\bar a_1x_3,x_4)\Rightarrow
I(x_1,x_2,x_3,x_4)$ is such a rule for $n=4$.

\begin{example}
  The grammar $G_2$  contains the following rules:
  \allowdisplaybreaks
  \begin{eqnarray*}
    S(x_1x_2) &\Rightarrow& I(x_1,x_2)\\
    I(x_1y_1x_2y_2,\varepsilon)&\Rightarrow & I(x_1,x_2),\,I(y_1,y_2)\\
    I(x_1y_1x_2,y_2)&\Rightarrow & I(x_1,x_2),\,I(y_1,y_2)\\
    I(x_1y_1,x_2y_2)&\Rightarrow & I(x_1,x_2),\,I(y_1,y_2)\\
    I(x_1,y_1x_2y_2)&\Rightarrow & I(x_1,x_2),\,I(y_1,y_2)\\
    I(\varepsilon,x_1y_1x_2y_2)&\Rightarrow & I(x_1,x_2),\,I(y_1,y_2)\\
    I(\alpha x_1\bar{\alpha},x_2)&\Rightarrow&I(x_1,x_2) \quad \alpha\in \Sigma_2\\
    I(\alpha x_1,\bar{\alpha}x_2)&\Rightarrow&I(x_1,x_2) \quad \alpha\in \Sigma_2\\
    I(\alpha x_1,x_2\bar{\alpha})&\Rightarrow&I(x_1,x_2) \quad \alpha\in \Sigma_2\\
    I(x_1\alpha,\bar{\alpha}x_2)&\Rightarrow&I(x_1,x_2) \quad \alpha\in \Sigma_2\\
    I(x_1\alpha,x_2\bar{\alpha})&\Rightarrow&I(x_1,x_2) \quad \alpha\in \Sigma_2\\
    I(x_1,\alpha x_2\bar{\alpha})&\Rightarrow&I(x_1,x_2) \quad \alpha\in \Sigma_2\\
    I(\varepsilon,\varepsilon)\\
  \end{eqnarray*}
\end{example}

  It is usual to present derivations as trees where nodes have the following
  form:
  \begin{center}
    \begin{prooftree}
      \vdash_G B_1(s_{1,1},\dots,s_{1,l_1}) \dots \vdash_G
      B_p(s_{p,1},\dots,s_{p,l_p})
      \justifies%
      \vdash_G A(s_1,\dots,s_n)
    \end{prooftree}
  \end{center}
  when from the derivations of the $\vdash_G B_i(s_{i,1},\dots,s_{i,l_i})$ we
  can derive $\vdash_G A(s_1,\dots,s_n)$ using a rule of $G$. When $p$ equals
  $0$, there is no derivation above the bar. This means that rules with no
  right-hand part form the leaves of these trees.

  Using this notation we can show as follows that $a_1 a_1 \bar{a}_2\bar{a}_1
  \bar{a}_1 a_2$ is in the language of $G_2$. We use colors to identify each
  letter and allow one to infer the rules used in the derivation:

  \begin{center}
    \begin{prooftree}
      \begin{prooftree}
      	 \begin{prooftree}
      		\begin{prooftree}
      			\justifies%
      			\vdash_{G_2}I(\varepsilon,\varepsilon)
      		\end{prooftree}
      		\justifies%
      		\vdash_{G_2} I(\textcolor{red}{a_1},\textcolor{red}{\bar a_1})
      	\end{prooftree}
        \begin{prooftree}
          \begin{prooftree}
            \begin{prooftree}
              \justifies%
              \vdash_{G_2}I(\varepsilon,\varepsilon)
            \end{prooftree}
            
            \justifies%
            \vdash_{G_2}I(\textcolor{green!50!black}{\bar{a}_2},
            \textcolor{green!50!black}{a_2})
          \end{prooftree}
          \justifies%
          \vdash_{G_2}I(\textcolor{blue}{a_1}
          \textcolor{green!50!black}{\bar a_2}, \textcolor{blue}{\bar a_1}\textcolor{green!50!black}{a_2})
        \end{prooftree}
        \justifies%
        \vdash_{G_2} I(\textcolor{red}{a_1}, \textcolor{blue}{a_1}
        \textcolor{green!50!black}{\bar a_2}\textcolor{red}{\bar
          a_1}\textcolor{blue}{\bar a_1} \textcolor{green!50!black}{a_2})
      \end{prooftree}
      \justifies%
      \vdash_{G_2}S(\textcolor{red}{a_1} \textcolor{blue}{a_1}
      \textcolor{green!50!black}{\bar a_2}\textcolor{red}{\bar
        a_1}\textcolor{blue}{\bar a_1} \textcolor{green!50!black}{a_2})
    \end{prooftree}
  \end{center}

\section{Main result}
\label{sec:main-result}
Our main theorem is that the language of $G_n$ is $O_n$. The inclusion of
$\mathcal{L}(G_n)$ into $O_n$ is obvious and the challenge consists in proving the 
converse inclusion. We prove actually a stronger statement:
\emph{$\vdash_{G_n}I(s_1,\dots,s_n)$ is derivable for every $s_1 \cdots s_n$
in $O_n$.} With this statement, the desired inclusion is immediate: for $w\in O_n$, set $s_1=w$ and 
$s_j=\varepsilon$ for all $j\geq 2$ and apply rule~\eqref{rule:init}.

A tuple $(s_1,\dots,s_n)$ is  \emph{reducible} if there are at least two compatible letters among the endpoints of the $s_j$'s; otherwise the tuple is
\emph{irreducible}. 
Repeated applications of rules of the form~\eqref{rule:const} allow to derive reducible tuples from irreducible ones. Irreducible tuples are dealt with the following ``decomposition lemma,'' which shows how a rule of the form~\eqref{rule:binary} is applied to derive an irreducible tuple from smaller elements in $O_n$. While rule~\eqref{rule:empty} provides the base case for the induction, this lemma is the main technical result towards the proof of \Cref{thm:o_n_n_mcfl}. It relies on \Cref{thm:multi-neckl}, a result stated and proved in \Cref{sec:necklace}, and formulated with the terminology of the necklace splitting theorem.

\begin{lemma}[Decomposition lemma]
	\label{lem:decomposition}
	Consider an irreducible tuple $(s_1,\dots, s_n)$ in $(\Sigma_n^{+})^n$, where each $s_j$ is of length at least $2$.
	If $s_1\cdots s_n$ belongs to $O_n$, then there exist two tuples
	$(u_1, u_3, \dots,u_{2n-1})$ and $(u_2, u_4, \dots,u_{2n})$ in $(\Sigma_n^{\ast})^n$ and integers \(0 = k_0 \leq k_1 \leq \cdots \leq k_n = 2n\) such that
	\begin{enumerate}[label=\textup{(\Roman*)}]
		\item \label{prop:nonempty-O_n}$u_1 u_3 \cdots u_{2n-1}$ and $u_2 u_4 \cdots u_{2n}$ are both nonempty elements of $O_n$, and
		\item\label{prop:interlacing} $s_j = u_{k_{j -1} + 1} ~ u_{k_{j -1} + 2} \cdots u_{k_j}$ for each \(j \in [n]\).
	\end{enumerate}
	In particular, \(s_1 \cdots s_n = u_1 u_2 \cdots u_{2n}\).
\end{lemma}

%

\begin{proof}
We distinguish the cases $n=1$ and $n\geq 2$.

We deal first with the case $n=1$. In that case, $\Sigma_n=\{a_1,\bar a_1\}$. As $s_1$ is irreducible, without loss of generality,
we may assume that $s_1 = a_1wa_1$. We consider prefixes $u'$ of $w$ of increasing length, from $|u'|=0$ to $|u'|=|w|$.
Since  $s_1$ belongs to $O_1$, the quantity $|a_1u'|_{a_1}-|a_1u'|_{\bar a_1}$ starts with the value $1$ 
and finishes with the value $-1$, and changes by steps of one unit.
There is therefore a prefix $u'$ of $w$ such that $a_1u'$ belongs to $O_1$. Setting $u_1=a_1u'$ and $u_2\in\Sigma_1^+$ such that $s_1=u_1u_2$
makes the job.


We deal now with the case $n\geq 2$. The proof of \Cref{thm:multi-neckl} builds explicitly $u_{\ell}$'s
satisfying the desired properties: property~\ref{prop:nonempty-O_n} is a consequence of the fact that each of $A$ and $B$ are
balanced; property~\ref{prop:interlacing} is a consequence of the fact that the endpoints of the $s_j$'s form cuts.
\end{proof}
From this the inclusion follows easily.

\begin{proposition}
  \label{prop:O_n_included_in_L_G_n}
  The judgment $\vdash_{G_n}I(s_1,\dots,s_n)$ is derivable for every $s_1\cdots s_n$
in $O_n$. In particular, we have $O_n\subseteq\mathcal{L}(G_n)$.
\end{proposition}
\begin{proof}
As mentioned above, the second part of the statement is a direct consequence of the first part. We proceed by induction
  on the pairs $(|s_1\cdots s_n|,e)$ ordered lexicographically, where $e$ is the number of $s_j$ equal to $\varepsilon$.


Suppose first that $(s_1,\dots,s_n)$ is reducible. A rule of the form~\eqref{rule:const} shows that we can derive 
the judgment from another judgment $\vdash_{G_n}I(s'_1,\dots,s'_n)$, with $|s'_1\cdots s'_n|<|s_1\cdots s_n|$. The induction hypothesis provides the conclusion.

Suppose now that $(s_1,\dots,s_n)$ is irreducible. Four cases are in order, distinguished according to the possible lengths of the $s_j$'s.

The first case is when at least one $s_j$ is of length $1$. Without loss of
  generality we may assume that $s_j = a_1$. As $s_1\cdots s_n$ is in $O_n$,
  there is a $k$ such that $s_k = v_1\bar a_1 v_2$ for some $v_1$ and $v_2$ in
  $\Sigma_n^\ast$. The other case being similar, we suppose that $j<k$.
Define $t_j = a_1$, $t_k = \bar a_1$, and $t_{\ell} = \varepsilon$ for $\ell\neq j,k$; define also 
$(u_1,\dots,u_n) = (s_1,\dots, s_{j-1},s_{j+1}, \dots ,s_{k-1},v_1,v_2, s_{k+1}, \dots
    , s_n)$. Using rule~\eqref{rule:empty} and a rule of the form~\eqref{rule:const}, 
we get that $\vdash_{G_n}I(t_1,\dots, t_n)$ is derivable. The judgment
$\vdash_{G_n}I(u_1,\dots,u_n)$ is derivable by induction. Then using a rule of
the form~\eqref{rule:binary} shows that $\vdash_{G_n}I(s_1,\dots, s_n)$ is derivable.
More precisely, we instantiate each variable $x_\ell$ with $t_\ell$ and each variable $y_\ell$ with $u_\ell$ in the following rule
  \begin{multline*}
  I(x_1y_1,\dots,x_{j-1}y_{j-1},x_{j}, y_j,\dots ,x_{k-2}y_{k-2}x_{k-1},y_{k-1}x_{k}y_k,\dots, x_n y_n)\Rightarrow \\ I(x_1,\dots, x_n), I(y_1,\dots,y_n ) \enspace .
  \end{multline*}

  The second case is when all $s_j$ are equal to $\varepsilon$. The conclusion follows from an application of rule~\eqref{rule:empty}.
  
  The third case is when some $s_j$ but not all are equal to $\varepsilon$ and no $s_j$ is of length $1$. There is
then an $j \in [n-1]$ such that either $s_j=\varepsilon$ and $|s_{j+1}|>1$, or $|s_j|> 1$ and
  $s_{j+1}=\varepsilon$. By symmetry, we suppose that $s_j=\varepsilon$
  and $|s_{j+1}|>1$. As $|s_{j+1}|>1$, we have that $s_{j+1} = v_1v_2$ with
  $v_1$ and $v_2$ in $\Sigma^+$. Define $(s'_1,\dots,s'_n) = (s_1,\dots,s_{j-1},v_1,v_2, s_{j+2}, \dots,s_n)$.
  The judgment
  $\vdash_{G_n}I(s'_1,\dots,s'_n)$ is derivable by induction (we have a smaller $e$). 
   The judgment  $\vdash_{G_n}I(\varepsilon,\dots,\varepsilon)$ is derivable from an application of rule~\eqref{rule:empty}.
Then using a rule of the form~\eqref{rule:binary} shows that $\vdash_{G_n}I(s_1,\dots, s_n)$ is derivable.
More precisely, we instantiate each variable $x_\ell$ with $s'_\ell$ and each variable $y_\ell$ with $\varepsilon$ in the following rule:
  \[
  I(x_1y_1,\dots, x_{j-1}y_{j-1},\varepsilon,x_i
    y_jx_{j+1}y_{j+1},\dots, x_ny_n)\Rightarrow I(x_1,\dots,
    x_n), I(y_1,\dots,y_n ) \enspace .
    \]
    
    The fourth case satisfies the conditions of \Cref{lem:decomposition}
    (``decomposition lemma''), which shows that $\vdash_{G_n}I(s_1,\dots, s_n)$
    is derivable by an application of a rule of the form~\eqref{rule:binary} and by induction.
\end{proof}

From this we can derive our main theorem.

\begin{theorem}
  \label{thm:o_n_n_mcfl} The language $O_n$ is an $n$-MCFL.
\end{theorem}

\begin{remark}
	Theorem~\ref{thm:multi-neckl} actually implies a version of 
	\Cref{lem:decomposition} that also holds for reducible tuples albeit only if 
	$n$ is at least $2$ and if we permit to decompose $(s_1, \ldots, s_n)$ in more versatile tuples.
	This corresponds to a slightly more liberal definition of the rules~\eqref{rule:binary}.
	More precisely, we may add to the rules~\eqref{rule:binary} for each \(j \in [n]\) the rule:
	\begin{equation*}\label{rule:binary_liberal}
	I(x_1 y_1, \ldots, x_{j-1}y_{j-1}, \text{\boldmath\bfseries{$y_j x_j$}},x_{j+1} y_{j+1}, \ldots, x_ny_n) \Rightarrow I(x_1, \ldots, x_n),\, I(y_1, \ldots, y_n) \enspace. \tag{\textasteriskcentered}
	\end{equation*}
	
	We have chosen to exclude rules~\eqref{rule:binary_liberal} 
	to emphasize the surprising simplicity of the rules~\eqref{rule:binary} 
	which allow to decompose every irreducible tuple.
	Moreover, we want to contrast the grammar we obtain with the one that
	Nederhof~\cite{Nederhof-2017-free-word-orders-and-MCFLs} conjectures to
	capture $O_n$. Nederhof proposes binary rules of the following form:
	\[A(w_1,\dots,w_n)\Rightarrow A(x_1,\dots,x_n),\,A(y_1,\dots,y_n)\] where
	$w_1\cdots w_n$ is obtained by shuffling the words $x_1\cdots x_n$ and
	$y_1\cdots y_n$, i.e., $|w_1\cdots w_n|=2n$, 
	removing all occurrences of \(y_j\)'s from $w_1 \cdots w_n$ yields $x_1 \cdots x_n$ and, 
	analogously, 
	removing all occurrences of \(x_j\)'s from $w_1 \cdots w_n$ yields $y_1 \cdots y_n$. 
	Furthermore, a $w_k$ may not contain 
	an occurrence of $x_jx_{j+1}$ or of $y_{j}y_{j+1}$ for some $j \in [n-1]$.
	The rules~\eqref{rule:binary} may be obtained from Nederhof's by forbidding the
	occurrence of $x_jx_{j+1}$ and of $y_{j}y_{j+1}$ not only in the $w_k$'s but
	rather in the combined word $w_1\cdots w_n$. 
	Note however that this additional restriction implies that rules~\eqref{rule:const} 
	need to allow the removal of compatible letters at arbitrary endpoints of the words of a tuple.
	The corresponding rules of Nederhof's grammar are less liberal.
	Were we to remove rules~\eqref{rule:const} and treat the elimination of compatible letters as
	in~\cite{Nederhof-2017-free-word-orders-and-MCFLs}, then we would need to add
	the rules~\eqref{rule:binary_liberal}. 
	Clearly, these rules form a strict subset of those proposed by Nederhof as 
	$w_k$'s are of length $2$ and at most one occurrence of $x_jx_{j+1}$ 
	is allowed in $w_1\cdots w_n$.
	In this sense, our grammar is simpler than Nederhof's.
	Notably, also Nederhof \cite[Sec.~5.3]{Nederhof-2017-free-word-orders-and-MCFLs} conjectures for \(O_3\)
	that some of the rules of his grammar are redundant.
\end{remark}


\section{Necklace splitting and proof of the decomposition lemma}\label{sec:necklace}

In this section, we prove a combinatorial theorem in the tradition of the
necklace splitting problem.  The theorem almost implies \Cref{lem:decomposition}, but while it does
not require irreducibility, it misses the property~\ref{prop:interlacing}. Its
proof however gives the construction for the case $n>1$ of
\Cref{lem:decomposition}.

This combinatorial theorem is formulated without
the terminology of languages and grammars so as to make it understandable easily
without background in this area. 
For readers who are more familiar to manipulating words, the vocabulary of the
necklace problem translates easily to that of words: necklaces become words and
beads become letters.

In the traditional version of the necklace splitting theorem, there is an open necklace with beads of $n$ different types, and an even number of beads of each type. The necklace splitting theorem ensures that such a necklace can always be split between two thieves with no more than $n$ cuts so that each thief gets the same amount of each type. The cuts have to leave the beads untouched. Here, we keep the same setting, except that a bead can be either ``positive'' or ``negative.'' The {\em amount} of a type in a necklace or in a collection of necklaces is the number of positive beads of this type minus the number of negative beads of this type. A necklace or a collection of necklaces is {\em balanced} if the amount of each type is zero.

\begin{theorem}\label{thm:multi-neckl}
Consider a collection of $n$ open necklaces with positive and negative beads. Suppose that there are $n\geq 2$ types of beads and that each of the necklaces has at least two beads. If the collection is balanced, then there is a way to cut the necklaces using at most $n$ cuts in total and partition the subnecklaces into two parts so that each part is balanced, gets at least one bead (and thus at least two), and is formed by at most $n$ subnecklaces.
\end{theorem}

The connection to \Cref{lem:decomposition} is as follows:

\begin{itemize}
\item The $n$ types of beads are the numbers $1, \dots,n$, the letter $a_i$
  representing a \emph{positive} bead of type $i$ and the letter $\bar{a}_i$
  representing a \emph{negative} bead of type $i$.
\item The $n$ open necklaces correspond to the words $s_1, \dots, s_n$.
\item The two parts of the at most $n$ subnecklaces are the tuples
  $(u_1,\dots,u_{2n-1})$ and $(u_2,\dots, u_{2n})$ in the lemma, which, as in the theorem,
  are required to be balanced and to contain at least one letter each.
\end{itemize}

The proof of \Cref{thm:multi-neckl} relies on the following proposition,
which is actually the traditional
necklace splitting theorem extended to negative beads (the relaxation of the parity condition of the number
of beads is standard; see, e.g.,~\cite[Section 5.1]{alon2006algorithmic}).

\begin{proposition}\label{prop}
Consider a necklace with positive and negative beads. Suppose that there are $n\geq 1$ types of beads. Then the necklace can be split between
two thieves with no more than $n$ cuts so that the amount of beads received by
the thieves differ by at most one unit for each type. It is moreover possible to choose
which thief receives an extra bead for each type which requires so.
\end{proposition}

The types requiring that one of the thieves receives an extra bead are precisely those whose total amount of beads is an odd number.

\begin{proof}[Proof of \Cref{prop}]
We proceed in two steps. In a first step, we prove a continuous version, in which we are allowed momentarily to locate cuts on beads themselves. In a second step, we show how to move the cuts so as to get a splitting with cuts leaving the beads untouched.\footnote{An alternative purely combinatorial proof is presented in \Cref{sec:alt-proof-prop2}.}

We identify the necklace with $[0,1]$ and the beads with intervals included in $[0,1]$, all of same length (this length is the inverse of the total number of beads in the necklace), and with disjoint interiors. Assume that these small intervals are all open. (This assumption is made to ease the proof, but actually whether these small intervals contain or not their boundaries does not matter.) Define
$$
g_i(x) \longmapsto \left \{ 
\begin{array}{ll} +1 & \text{if $x$ is in a interval corresponding to a positive bead of type $i$.} \\ 
-1 & \text{if $x$ is in a interval corresponding to a negative bead of type $i$.} \\ 
0 & \text{otherwise.}
\end{array} \right.
$$
According to the Hobby--Rice theorem~\cite{hobby1965moment}, there are points $0=x_0 < x_1 < \cdots < x_r < x_{r+1}=1$ with $r \leq n$ such that for all $i \in [n]$
$$\sum_{j=1}^{r+1} (-1)^j \int_{x_{j-1}}^{x_j}g_i(u)\operatorname{d}u=0\enspace.$$ (Here, we take the formulation given by Pinkus~\cite{pinkus1976simple}.) The $r$ points $x_1,\ldots,x_r$ can be interpreted as cuts. The intervals $(x_{j-1},x_j)$ with $j$ odd are given to one thief, and the intervals $(x_{j-1},x_j)$ with $j$ even are given to the other thief. Each thief gets the same amount of each type. The only problem is that some cuts may be located on beads.

The second step of the proof consists in explaining how to move these cuts so that none of them touch the beads anymore, without creating a difference of more than one unit between the amounts received by the thieves for each type. We can make that each bead is cut at most twice since moving two cuts inside a bead in the same direction and by the same distance does not change the amounts received by the thieves. If a bead is cut twice, we can similarly move the two cuts until one of them is located between two beads. (Doing this, we can have several cuts located at the same position between two beads, but this is not an issue.) So, we can assume that each bead is cut at most once and that the two parts of a cut bead go to distinct thieves. If a type has two beads or more touched by a cut, then we can move two cuts so that one at least reaches a position between two beads, without changing the amounts received by each thief. Thus, we can assume that each type is cut at most once. We finish the proof by noting that if a type has a bead that is cut, it means that each thief received a non-integral amount of the corresponding type, i.e., a half-integer, and moving the cut arbitrarily leads to the desired splitting.
\end{proof}




\begin{proof}[Proof of \Cref{thm:multi-neckl}]
Denote by $s_1,\ldots,s_n$ the $n$ necklaces. We assume
that there are no two beads located at the endpoints of some necklaces 
and that are of the same type but of opposite signs, 
for otherwise there would be an easy solution: cut these two beads from the necklaces, form a balanced
part with them, and leave the rest of the remaining beads for the second balanced
part. The number of cuts would then be $2\leq n$ (or, if these two beads formed
a necklace of their own, the solution would need no cut), and the number of
subnecklaces in the parts would be $2$ and $n$ (or $1$ and $n-1$).
Making this assumption corresponds to considering only the irreducible case in the terminology of \Cref{thm:o_n_n_mcfl}.

A natural idea would be to apply a result like \Cref{prop} to the ``big'' necklace 
$s=s_1\cdots s_n$ obtained by appending the necklaces in their index order. 
The first issue with this idea is that one thief might get nothing. This can 
easily be dealt with as done below. The second issue is that, even though 
there are $2n$ subnecklaces in total, one thief might get more than $n$ subnecklaces.\footnote{
In essence, this is the reason why Ho~\cite{Thewordproblemofnisamultiplecontextfreelanguage} could only show
that $O_n$ is an $\left(8\left\lfloor \frac{n+1}{2} \right\rfloor-2\right)$-MCFL.}
Instead we consider another big necklace, which we will denote by $s'$ and which we define now.

We start by defining $s'_1$ to be $s_1$ from which the left-most bead has been removed,
and $s'_n$ to be $s_n$ from which the right-most bead has been removed. Note
that without loss of generality, we can assume that the left-most bead of $s_1$
is a positive bead of type $1$ and that the right-most bead of $s_n$ is a positive bead of type $1$ or $n$.
We thus consider two cases:
\begin{enumerate}[label=(\roman*)]
\item\label{case:11} The left-most bead of $s_1$ and the right-most bead of
  $s_n$ are both positive beads of type $1$.
\item\label{case:1n} The left-most bead of $s_1$ is a positive bead of type $1$ and the right-most
  bead of $s_n$ is a positive bead of type $n$.
\end{enumerate}

The condition of the theorem ensures that neither $s'_1$ nor $s'_n$
are empty. Consider the ``big'' necklace $s' = s'_1\bar s_2s_3\bar s_4\cdots$. If $n$ is even, the
big necklace $s'$ ends with $\bar s'_n$, and if $n$ is odd, it ends with $s'_n$.
(Given a sequence $t$ of beads, the notation $\bar t$ means $t$ where all positive beads become negative and conversely, without changing their types.) 
According to \Cref{prop}, there is a splitting of $s'$ into $n+1$
subnecklaces $t_1,\ldots,t_{n+1}$ (completing with zero-length subnecklaces if necessary) such that we have for $i \in [n]$ in Case~\ref{case:11} and for $i\in\{2,\ldots,n-1\}$ in Case~\ref{case:1n}
\begin{align}
  \ds{\sum_{k\text{ odd}}\mu_i(t_k)} & = \ds{\sum_{k\text{ even}}\mu_i(t_k)} \enspace,\label{eqi}
\end{align}
and such that in Case~\ref{case:1n}
\begin{align}
  \ds{\sum_{k\text{ odd}}\mu_1(t_k)} & = \ds{\sum_{k\text{ even}}\mu_1(t_k)-1} \enspace,\label{1teq1} \\
  \ds{\sum_{k\text{ odd}}\mu_n(t_k)} & = \ds{\sum_{k\text{ even}}\mu_n(t_k)+1} \enspace. \label{1teqn}
\end{align}
Here, we denote by $\mu_i(t)$ the amount of beads of type $i$ in the subnecklace $t$. We remind the reader that the amount of beads is the number of positive beads minus the number of negative beads of this type.

We interpret now the endpoints of the subnecklaces $t_k$ as cuts of the ``big'' necklace $s$ (the one formed by the original necklaces $s_j$). 
Together with the endpoints of the $s_j$, we get a splitting of $s$ into $2n$ subnecklaces $u_1,\ldots,u_{2n}$
(in this order). Some $u_{\ell}$ can be zero-length subnecklaces.
We make two parts: a part $A$ formed by the 
$u_{\ell}$ with an odd $\ell$, and a part $B$ formed by the $u_{\ell}$ with an even $\ell$.
Remark two things:
\begin{itemize}
\item Since the left-most bead of $s_1$ belongs to $u_1$ and the right-most bead of $s_n$ belongs to $u_{2n}$, each part contains at least one bead.
\item The number of subnecklaces  $u_{\ell}$ is the same in each part, and thus equal to $n$.
\end{itemize}

We finish the proof by checking that both $A$ and $B$ are balanced. 
We denote by
$q^A_i$ and $q^B_i$ the amount of beads of type $i$
in $A$ and $B$, respectively. Since the
original collection is balanced, we have $q^A_i+q^B_i=0$. The end of the proof consists simply in checking that we have also $q^A_i-q^B_i=0$, which implies then immediately that $q^A_i=q^B_i=0$, as desired.

Each bead $x$ belongs to exactly one $s_j$. 
Moreover, apart from the two beads at the endpoints of $s$, any bead $x$ (resp. $\bar x$) belongs also to exactly one $t_k$ when $j$ is odd (resp. even). It is immediate to check
that $j+k-1$ and the index $\ell$ of the $u_{\ell}$ to which $x$ belongs have the same parity. Thus if $j+k-1$ is odd, then $x$ belongs to $A$,
and if it is even, then $x$ belongs to $B$. We have
\[
q^A_1 = 1 + \sum_{j,k\text{ odd}} \mu_1(s_ j\cap t_k) - \sum_{j,k\text{ even}}  \mu_1(\bar s_ j\cap t_k)  ~~ \mbox{and} ~~
q^B_1 =  \delta_{\text{\ref{case:11}}} + \sum_{\substack{j\text{ odd} \\ k\text{ even}}} \mu_1(s_ j\cap t_k) - \sum_{\substack{j\text{ even} \\ k\text{ odd}}}  \mu_1(\bar s_ j\cap t_k) \enspace ,
\]
and for $i \neq 1$, we have 
\[
q^A_i = \sum_{j,k\text{ odd}} \mu_i(s_ j\cap t_k) - \sum_{j,k\text{ even}}  \mu_i(\bar s_ j\cap t_k)  \quad \mbox{and} \quad
q^B_i = \delta^{i = n}_{\text{\ref{case:1n}}} + \sum_{\substack{j\text{ odd} \\ k\text{ even}}} \mu_i(s_ j\cap t_k) - \sum_{\substack{j\text{ even} \\ k\text{ odd}}}  \mu_i(\bar s_ j\cap t_k) \enspace ,
\]
where $\delta_{\text{\ref{case:11}}} \in \{0,1\}$ and takes the value $1$ only if we are in Case~\ref{case:11}, and where $\delta^{i=n}_{\text{\ref{case:1n}}} \in \{0,1\}$ and takes the value $1$ only if we are in Case~\ref{case:1n} and $i=n$.

For the beads of type $1$, we have
\[
\begin{array}{rcl}
  q^A_1- q^B_1 & = & \ds{1 -  \delta_{\text{\ref{case:11}}} + \sum_{k=1}^{n+1}\sum_{j\text{ odd}}(-1)^{k+1}\mu_1(s_j\cap t_k) +
   \sum_{k=1}^{n+1}\sum_{j\text{ even}}(-1)^{k+1}\mu_1(\bar s_j\cap t_k)} \smallskip\\
            & = & \ds{1 -  \delta_{\text{\ref{case:11}}} + \sum_{k=1}^{n+1}(-1)^{k+1}\mu_1(t_k)}\smallskip \\
            & = & 0  \enspace ,
\end{array}
\]
where the last equality is a consequence of Equation~\eqref{eqi} in Case~\ref{case:11} and of Equation~\eqref{1teq1} in Case~\ref{case:1n}. For the beads of type $i \neq 1$, we have
\[
\begin{array}{rcl}
  q^A_i- q^B_i & = & \ds{- \delta^{i = n}_{\text{\ref{case:1n}}} + \sum_{k=1}^{n+1}\sum_{j\text{ odd}}(-1)^{k+1}\mu_i(s_j\cap t_k) +
   \sum_{k=1}^{n+1}\sum_{j\text{ even}}(-1)^{k+1}\mu_i(\bar s_j\cap t_k)} \smallskip\\
            & = & \ds{-  \delta^{i = n}_{\text{\ref{case:1n}}}+ \sum_{k=1}^{n+1}(-1)^{k+1}\mu_i(t_k)}\smallskip \\
            & = & 0  \enspace ,
\end{array}
\]
where the last equality is a consequence of Equation~\eqref{eqi}, except when $i=n$ and we are in Case~\ref{case:1n}, where we use of Equation~\eqref{1teqn} instead.
\end{proof}

\section{Alternate combinatorial proofs of Proposition~\ref{prop} and Lemma~\ref{lem:decomposition}}
\label{sec:combinatorial-proofs}

All statements deal with discrete objects and properties. It is thus desirable
from a logical point of view that all proofs stay in the ``discrete world'' if possible. Similarly to the traditional proof of the necklace splitting theorem, the proof of \Cref{prop} we gave in \Cref{sec:necklace} is relying on some ``continuous'' notions. P\'alv\H{o}lgyi~\cite{palvolgyi2009combinatorial} showed how such proofs are amenable to the discrete world by using a combinatorial counterpart of the Borsuk--Ulam theorem.

We adapt here the approach proposed by P\'alv\H{o}lgyi to provide a purely combinatorial proof of \Cref{prop}, also relying on Tucker's lemma. Moreover, we show how the more general Ky Fan lemma can actually provide a direct combinatorial proof of \Cref{lem:decomposition} itself.

\subsection{Combinatorial tools}

We start by introducing some notation.
Let \(\Bool = \{-1, 1\}\). For brevity, we may write \(+\) instead of \(+1\) and \(-\) instead of \(-1\).
Let \(\mO = \{-1,0, 1\}\). We define the partial order \(\prec\) on \(\mO\) where \(0 \prec 1\) and \(0 \prec -1\), or, equivalently, \(b \preceq b'\) if \(b = 0\) or \(b = b'\).
For every \(m \in \mathbb{N}\), we lift \(\prec\) to \(\mO^m\) where for \(x, y \in \mO^m\) we have \(x \preceq y\) if for all \(i \in [m]\), \(x(i) \preceq y'(i)\). Given a string \(x\), we denote by \(x(i)\) the \(i\)-th letter of \(x\). We also denote by $-x$ the
string obtained from $x$ by replacing $1$'s by $-1$'s and $-1$'s by
$1$'s (and the $0$'s are left unchanged).

Now the Ky Fan lemma can be stated as follows:
\begin{lemma}[Octahedral Ky Fan lemma, {\cite[Lemma 2]{CHEN20111062}}]\label{lemma:ky-fan} 
	Let \(\lambda\colon \mO^m \setminus 0^m \to \{\pm1, \ldots, \pm q\}\) such that 
	\begin{enumerate}[label=\textup{(\roman*)}]
		\item\label{tucker:1} \(\lambda(-x) = -\lambda(x)\) for every \(x\).
		\item\label{tucker:2} \(\lambda(x) + \lambda(y) \neq 0\) for every \(x \preceq y\).
	\end{enumerate}
	Then there is at least one positively alternating \(m\)-chain, i.e.,
	there are \(x_1, \ldots, x_m \in \mO^m \setminus 0^m\) and \(j_1, \ldots, j_m \in [q]\) such that 
	\(x_1 \preceq \cdots \preceq x_m\), \(1 \leq j_1 < j_2 < \cdots < j_m\), and 
	\(\lambda(\{x_1, \ldots, x_m\}) = \{j_1, -j_2, j_3, \ldots, (-1)^{m-1} j_m\}\). 
	In particular, \(q \geq m\).
\end{lemma}
(The statement we provide here is actually a special case of the original Ky Fan lem\-ma~\cite{fan52:_gener_tucker_combin_lemma_topol_applic} when the simplicial complex is the first barycentric subdivision of the octahedron.) The octahedral Tucker lemma~\cite{tucker_lemma_46} is actually the same statement without the existence of the chain, but still with the inequality $q \geq m$.

We explain now how strings in $\mO^m$ relate to decompositions of strings in $\Sigma_n^m$. We consider here all decompositions of a string $w \in \Sigma_n^m$ into two tuples $(u_1, u_3,\dots)$, $(u_2,u_4,\dots)$ of strings in $\Sigma_n^{+}$ such that $w = u_1u_2\cdots$. Such a decomposition is described as a string $x \in \Bool^m$ where each maximal segment of consecutive $-1$'s, as well 
as each maximal segment of consecutive $+1$'s, corresponds to one $u_j$.  Each sign change in $x$ corresponds thus to a cut in $w$ and to a transition from a $u_j$ to $u_{j+1}$. A string $x \in \mO^m$, containing some $0$'s, can be interpreted as an underspecified decomposition, which can be completed into different decompositions by replacing the $0$'s with $-1$'s or $+1$'s.


\subsection{Combinatorial proof of Proposition~\ref{prop}}
\label{sec:alt-proof-prop2}

We prove \Cref{prop} in a way that is similar
to~\cite{palvolgyi2009combinatorial}. 

For $x$ in $\mathbb{B}^\ast$ we let $\operatorname{alt}(x)$ be the number of sign
alternations in $x$. We extend the function $\operatorname{alt}$ to $\mathbb{O}^\ast \setminus 0^\ast$
as follows: $\operatorname{alt}(x) = \max\{\operatorname{alt}(y) \mid y\in\mathbb{B}^\ast,\,
x\preceq y\}$. Notice that:
\begin{itemize}
\item $\operatorname{alt}(-x) = \operatorname{alt}(x)$ for every $x$.
\item $\operatorname{alt}(x)\geq \operatorname{alt}(y)$ for every $x\preceq y$.
\end{itemize}
For $x$ in $\mathbb{O}^\ast \setminus 0^\ast$, we define $\operatorname{sign}(x) \in\{-1,1\}$ as
the first letter of some $y$ in $\mathbb{B}^\ast$ so that $x\preceq y$ and
$\operatorname{alt}(x) = \operatorname{alt}(y)$. The function $\operatorname{sign}(x)$ is well
defined as when $x$ is of the form $0^k \kappa x'$, with $\kappa \in \mathbb{B}$, the only way to replace the
first $k$ zeroes so as to maximize $\operatorname{alt}$ consists in changing the sign at
each position. 

From now on we fix the necklace $w$ and assume it is of length $m$ (i.e., $w$ is
in $\Sigma_n^m$).

We let $E_{\kappa,i}(x)$ for $x$ in $\mathbb{O}^m \setminus 0^m$ and $\kappa \in
\{{-1},\,1\}$ be the amount of beads of type $i$ in $w$ that are aligned with
the symbol $\kappa$ in $x$. We say that $x$ is $\kappa i$-\emph{unbalanced} if 
$E_{\kappa,i}(y)>E_{-\kappa,i}(y)$ for every $y$ such that $x\preceq y$. We
define $\unb(x)$ to be $\kappa i$ where $i$ is the smallest $i$ so that
$x$ is $\kappa i$-unbalanced. When no such $i$ exists, we let $\unb(x) =0$. We have
\begin{itemize}
	\item $\unb(-x) = -\unb(x) $ for every $x$. 
	\item $|{\unb(x)}| \geq |{\unb(y)}| > 0$ for every $x \preceq y$ if $\unb(x) \neq 0$.
\end{itemize}  
Notice that $|{\unb(x)}| = |{\unb(y)}|$ implies	$\unb(x) =\unb(y)$
for every $x\preceq y$.
Finally we define $\lambda$ as the function from $\mathbb{O}^m$ to $\{-m+1,\dots,
0,\dots,m-1\}$:
\[
  \lambda (x) = \left\{
    \begin{array}{ll}
      \operatorname{sign}(x)\operatorname{alt}(x)&\text{ when } \operatorname{alt}(x) > n\enspace .\\
      \unb(x) & \text{ when } \operatorname{alt}(x) \leq n \enspace .
    \end{array}
  \right .
\]
Because of the properties of $\operatorname{alt}$, $\operatorname{sign}$ and $\unb$,
had it no zero, the function $\lambda$ would satisfy the properties
\ref{tucker:1} and \ref{tucker:2} of the octahedral Tucker lemma. However, then
the conclusion is not satisfied, thus $\lambda$ must have a zero. Take $x$ so
that $\lambda(x) = 0$. We must have that $\operatorname{alt}(x)\leq n$. As
$\unb(x)=0$, it is possible, for each type $i$, to replace $0$'s
in $x$ with $1$'s or $-1$'s so that we eventually obtain $y$ verifying $E_{1,i}(y) = E_{-1,i}(y)$ for every
$i$, 
and such that the remaining $0$'s in $y$ are aligned with
at most one unassigned bead of type $i$. We can then choose to replace these
$0$'s with either $1$ or $-1$ depending on how we wish to treat the extra bead of a
given type. As $\operatorname{alt}(x)$ is smaller than $n$, then no matter how we have
conducted the previous changes, we have obtained a way to split the necklace
with at most $n$ cuts so that each part contains the same amount of each
type, with the extra beads (when the amount is odd) being distributed in
any possible way.

\begin{remark}
	The definition of $\unb$ is where
	the proof departs from the one of~\cite{palvolgyi2009combinatorial}. Indeed P\'alv\H{o}lgyi's
	proof only considers positive beads and it is then enough to consider that $x$
	is unbalanced when one of the thief receives more than half of the beads of
  some type. In that case the split remains unbalanced for every $y$ so that $x\preceq
	y$. This property is no longer true with negative beads.
	Our definition of $\unb$ is taking this into account by imposing that
	being unbalanced is closed under $\preceq$.
\end{remark}

\subsection{Combinatorial proof of Lemma~\ref{lem:decomposition}}
Lastly, we present a direct combinatorial proof of \Cref{lem:decomposition} that avoids 
\Cref{thm:multi-neckl}. The method is again inspired by P\'alv\H{o}lgyi~\cite{palvolgyi2009combinatorial}, 
but instead of constructing a function \(\lambda\) that contradicts the Tucker lemma, if it has 
no zero, we construct a pair of functions \(\lambda_+\) and \(\lambda_-\), for which the chains guaranteed by 
the Ky Fan lemma cannot exist for irreducible tuples.

In contrast to \Cref{sec:alt-proof-prop2}, we do not consider \(x \in \Bool^m\) to represent the decomposition
of a single string but instead $x$ shall represent the decomposition of an 
irreducible tuple $(s_1,\dots, s_n)$ in $(\Sigma_n^{+})^n$ 
where the word $s_1 \cdots s_n$ is in $O_n$ and each \(s_j\) is of length at least \(2\).
To account for the arity of the tuple, we need to generalize the notion of \(\malt\).
Denote the word \(s_1 \cdots s_n\) by \(s\) and the length of \(s\) by \(m\). 
\begin{figure}[b]
	\centering
	\begin{tikzpicture}
		\matrix[matrix of math nodes, ampersand replacement=\&, every node/.append style={font=\strut}] (m) {
			\& a_1 \& a_2 \& \bar{a}_2 \& a_3 \& \bar{a}_2 \& \bar{a}_3 \& a_1 \& \bar{a}_1 \& \bar{a_1} \& \bar{a}_2 \\
			x = \&  +  \&  -  \&   -       \&   - \&  -        \&   -       \&  -  \&    +      \&   -       \&    -      \\
			\malt(x) = \& \& \& \& \& \& \& \& \& \& \& = 5	 \\
			 (    \& u_1 \& \& u_2 \& \& \& u_4  \& ~ \& u_5 \& u_6 \& \& ) \\
		};
		\begin{scope}[b/.style={decorate,decoration={brace,amplitude=10pt}, shorten >= 4pt, shorten <= 4pt}]
			\draw[b] (m-1-2.north west) to node[above=8pt] {$s_1$} (m-1-5.north east);
			\draw[b] (m-1-6.north west) to node[above=8pt] {$s_2$} (m-1-8.north east);
			\draw[b] (m-1-9.north west) to node[above=8pt] {$s_3$} (m-1-11.north east);
		\end{scope}
		\draw[-latex] (m-2-2.south) to[bend right=40] node[below] {+1} (m-2-3.south);
		\draw[-latex] (m-2-5.south) to[bend right=40] node[below] {+2} (m-2-6.south);
		\draw[-latex] (m-2-8.south) to[bend right=40] node[below] {+1} (m-2-9.south);
		\draw[-latex] (m-2-9.south) to[bend right=40] node[below] {+1} (m-2-10.south);
		\node[anchor=base] at ($(m-4-4.base)!.4!(m-4-7.base)$) {$u_3,$};
		\node[anchor=base] at ($(m-4-8.base)!.5!(m-4-9.base)$) {$,$};
		\begin{scope}[on background layer, plus/.style={pattern=north east lines, draw=gray!30, pattern color=gray!30}, minus/.style={pattern=north west lines, draw=gray!30, pattern color=gray!30}]
			\draw[plus, ] (m-4-2.120)  -- (m-1-2.south west) --  (m-1-2.south east) --  (m-4-2.60) --cycle;
			\draw[minus,] (m-4-4.120)  -- (m-1-3.south west) --  (m-1-5.south east) --  (m-4-4.60) --cycle;
			\draw[minus,] (m-4-7.120)  -- (m-1-6.south west) --  (m-1-8.south east) --  (m-4-7.60) --cycle;
			\draw[plus, ] (m-4-9.120)  -- (m-1-9.south west) --  (m-1-9.south east) --  (m-4-9.60) --cycle;
			\draw[minus,] (m-4-10.120) -- (m-1-10.south west) -- (m-1-11.south east) -- (m-4-10.north east) --cycle;
		\end{scope}
	\end{tikzpicture}
	\caption{A word \(s_1s_2s_3\), a decomposition \(x\) with the corresponding alignment of \(s_j\)'s to \(u_p\)'s (in particular, \(u_3 = \varepsilon\)), and the calculation of \(\malt(x)\).}
	\label{fig:maltXinterlace}
\end{figure}
Factorize \(x\) such that \(x = \kappa^{\ell_1}_1 \kappa^{\ell_2}_2 \cdots \kappa^{\ell_k}_k\) with \(\kappa_1, \ldots, \kappa_k \in \Bool\), \(\kappa_p = - \kappa_{p+1}\) and \(\ell_p \in \mathbb{N}\) 
be such that for each \(j \in [n]\), there is \(q_j \in \{0, \ldots, k\}\) with \(\sum_{p \in [j]} |s_p| = \sum_{p \in [q_j]} \ell_p\), and \(k\) is minimal. Now, if \(k = 2n\), then \(x\) describes a decomposition of $(s_1, \ldots, s_n)$ in two tuples each of size \(n\).
Precisely, these tuples are $(u_1, u_3, \ldots, u_{2n-1})$ and $(u_2, u_4, \ldots, u_{2n})$ 
where the length of $u_j$ is $\ell_j$ for each $j \in [2n]$ 
and \(s_j = u_{q_{i-1} + 1} u_{q_{i-1} + 2} \cdots u_{q_j}\) for each \(j \in [n]\).
If \(k < 2n\), we may add \(\kappa_q\)'s with \(\ell_q = 0\) to increase \(k\) to \(2n\) and proceed similarly. 
We call the value \(k-1\) the number of \emph{sign alternations of \(x\)} and write \(\malt(x) = k - 1\). 
Hence, in the following we will search for \(x\) with \(\malt(x) \leq 2n - 1\).
Crucially, we assume mandatory sign alternations between the \emph{neighboring endpoints} of \(s_j\) and \(s_{j+1}\): if the last position of \(s_j\) and the first position of \(s_{j+1}\) are signed equally in \(x\), then some \(\kappa_p\) with \(\ell_p = 0\) occurs in our factorization of \(x\). In this case $2$ sign alternations are accounted for the transition between neighboring endpoints.
In order to maximize \(\malt\), we may choose \(x\) to be strictly alternating except for the \(n-1\) transitions between neighboring endpoints. In consequence \(\malt(x) \leq m - (n-1) + 2(n-1) - 1 = m + n - 2\).

\Cref{fig:maltXinterlace} gives an example: Note that \(\malt(x) < 2\cdot3\) and that the subwords in \(+\)-labeled and \(-\)-labeled positions are both in \(O_3\).
Thus, \(x\) describes the decomposition of \( (a_1  a_2  \bar{a}_2  a_3,  \bar{a}_2  \bar{a}_3  a_1 , \bar{a}_1  \bar{a_1}  \bar{a}_2) \) to two strictly smaller \(3\)-tuples \((a_1, \varepsilon, \bar{a}_1)\) and \((a_2  \bar{a}_2  a_3,  \bar{a}_2  \bar{a}_3  a_1 , \bar{a_1}  \bar{a}_2 )\), for each of which the word obtained by concatenating the tuple's components is in \(O_3\).

Again, we consider words in \(\mO^m \setminus 0^m\) which may have \emph{unsigned} positions, i.e., positions labeled with \(0\). 
In this case the decomposition is only partially determined. We define the function \(h\colon \mO^m \setminus 0^m \to [m + n - 2]\) by 
\(h(x) = \max_{y \in \Bool^m \colon x \preceq y} \malt(y)\).
Observe that \(\malt(x) = \malt(-x)\), hence, for each \(x \in \mO^m \setminus 0^m\): \(h(x) = h(-x)\).
Also, for each \(x, y \in \mO^m \setminus 0^m\) with \(x \preceq y\), we have \(h(x) \geq h(y)\).

Denote by \(H(x)\) the set \(\arg\max_{y \in \Bool^m\colon x \preceq y} \malt(y)\).
All elements of \(H(x)\) may be obtained with the following algorithm: 
\begin{enumerate}[label = (\roman*)]
	\item Choose a signed position \(p\) in \(x\) with an unsigned neighbor \(p' \in \{p-1, p+1\}\).\label{step:1}%
	\item If \(p\) and \(p'\) are neighboring endpoints (in this case we call \(p\) and \(p'\) \emph{internal endpoints}), set \(x'(p') = x(p)\); otherwise set \(x'(p') = - x(p)\). For each \(q \neq p'\), set \(x'(q) = x(q)\).
	\item If \(x' \in \Bool^m\) we are done, otherwise, recursively apply the algorithm to \(x'\).
\end{enumerate}
This algorithm may yield words \(y\), \(y'\) where \(y(p) \neq y'(p)\) if \(p\) is an unsigned position between two signed positions of \(x\). However, for all position \(p'\) smaller than the smallest signed position of \(x\), 
the signs \(y(p')\) and \(y'(p')\) are equal for all words \(y,y'\) in \(H(x)\). 
Hence, we may define a function \(\sign\colon \mO^m \setminus 0^m \to \Bool\) to assign to \(x\) the first symbol of some \(y \in H(x)\). 
Moreover, for each \(x \in \mO^m \setminus 0^m\): \(\sign(x) = -\sign(-x)\).

To prove \Cref{lem:decomposition}, we need to show that there is \(x \in \Bool^m\) with \(\malt(x) < 2n\) and that 
both words \(u_1u_3\cdots\) and \(u_2u_4\cdots\) described by \(x\) are in \(O_n\) and non-empty. 
To assure non-emptiness, we just need that both signs occur in \(x\). To ensure that the words are in \(O_n\), 
we rephrase the notions \(E_{\kappa, i}\) and \(\unb\) from \Cref{sec:alt-proof-prop2} in terms of tuples of words instead of necklaces. Now \(E_{\kappa, i}(x)\) is formulated as the difference of the amount of \(a_i\)'s and \(\bar{a}_i\)'s aligned with \(\kappa\) in \(x\). Formally: \[E_{\kappa,i}(x) = |\{p \mid s(p) =
  a_i \land x(p) = \kappa\}|-|\{p \mid s(p) = \bar{a}_i \land x(p) =
  \kappa\}|\enspace.\] Observe that \(E_{+,i}(x) + E_{-,i}(x) = 0\) for \(s \in
O_n\) and \(x \in \Bool^m\).
%
%
We define two functions \(\lambda_+\) and \(\lambda_-\), where  
for each \(b \in \Bool\), \(\lambda_b\colon \mO^m \setminus 0^m \to \{\pm1, \ldots, \pm m\} \cup \{0\}\) is such that
	\newsavebox{\mycases}
	\reqnomode
	\begin{alignat}{2}
		\sbox{\mycases}{$\displaystyle \lambda_b(x) =\left\{%
			\begin{array}{@{}c@{}}%
				\vphantom{\sign(x) \cdot (h(x) - n + 2)\ \sign(x) \cdot (h(x) - n + 2)  \quad \tif h(x) \geq 2n}\\%
				\vphantom{\unb(x)  \quad \tif h(x) < 2 n \land |x|_{+} > 0 \land |x|_{-} > 0} \\ 
				\vphantom{b \cdot (-1) \cdot b' \cdot (n + 1)  \tif h(x) < 2n \land \exists b' \in \Bool \colon |x|_{b'} = 0  }\\
				\vphantom{\cdot}%
			\end{array}%
			\right.\kern-\nulldelimiterspace$}
		\raisebox{-.5\ht\mycases}[0pt][0pt]{\usebox{\mycases}}
		&\sign(x) \cdot (h(x) - n + 2) && \quad \tif h(x) \geq 2n\enspace. \label{case:hg2n} \tag{C1}\\
		&b \cdot (-1) \cdot b' \cdot (n + 1) && \quad \tif h(x) < 2n \land \exists b' \in \Bool \colon |x|_{b'} = 0 \enspace.\label{case:singlesigned} \tag{C2}\\
		&{\unb(x)}&& \quad \tif h(x) < 2 n \land |x|_{+} > 0 \land |x|_{-} > 0 \enspace.\label{case:unbal} \tag{C3}
	\end{alignat}
	\leqnomode
If there are \(x \in \mO^m \setminus 0^m\) and \(b \in \Bool\) such that \(\lambda_b(x) = 0\), 
then \Cref{lem:decomposition} holds: Note that case \ref{case:unbal} of \(\lambda_b\) applies. 
Thus, for each \(i \in [n]\) there is \(y \in \Bool^m\) such that \(E_{+,i}(y) = E_{-,i}(y) = 0\).
The positions relevant to balance \(a_i\) and \(a_j\) where \(i \neq j\) are distinct. 
Hence, we choose \(\hat{y} \in \Bool^m\) with \(x \preceq \hat{y}\) such that all symbols are balanced. 
Since \(\hat{y} \succeq x\) we have that \(\malt(\hat{y}) \leq h(x) < 2n\), \(|\hat{y}|_{+} > 0\), and \(|\hat{y}|_- > 0\). Thus, \(\hat{y}\) encodes a decomposition with the desired properties.
	
Otherwise, if \(\lambda_+\) and \(\lambda_-\) have no zero, we want to apply the Ky Fan lemma. 
Clearly, \(\lambda_b\) satisfies property \ref{tucker:1} of the octahedral Ky Fan lemma if cases \ref{case:hg2n} and \ref{case:unbal} apply. For case \ref{case:singlesigned}, note that \(|x|_{+} = 0 = |{-x}|_{-}\) implies 
\[-\lambda_b(x) = -(b \cdot (-1) \cdot (+1) \cdot (n+1)) = (b \cdot (-1) \cdot (-1)\cdot (n+1)) = \lambda_b(-x) \enspace. \]
It remains to show property \ref{tucker:2}, i.e., that the sum \(\lambda_b(x) + \lambda_b(y)\) is not zero if \(x \preceq y\) and \(\lambda_b(x) \neq 0 \neq \lambda_b(y)\). 
Assume \(\lambda_b(x) + \lambda_b(y) = 0\). By the definition of \(\lambda_b\), there are three cases:
	\begin{enumerate}
		\item[\ref{case:hg2n}:] \(h(x) = h(y) \geq 2n\) and \(\sign(x) = -\sign(y)\). But then \(H(x) \supseteq H(y)\) and, thus, \(\sign(x) = \sign(y)\), a contradiction.
		\item[\ref{case:singlesigned}:] \(h(x), h(y) < 2n\) and either \(|x|_+ = 0 = |y|_-\) or \(|x|_- = 0 = |y|_+\).
		W.l.o.g.\@ assume that \(|x|_- = 0 = |y|_+\). Then there is \(k \in [m]\) such that \(x(k) = +\). 
		But \(x \preceq y\) implies \(y(k) = +\), a contradiction to \(|y|_+ = 0\).
		\item[\ref{case:unbal}:] \(h(x), h(y) < 2n\), \(|x|_+, |y|_+ > 0\), \(|x|_-,|y|_- > 0\) and \(\unb(x) = -\unb(y)\). 
		This contradicts properties of \(\unb\) shown in \Cref{sec:alt-proof-prop2}: for every \(x \preceq y\) with \(\unb(x) \neq 0\) we have $|{\unb(x)}| \geq |{\unb(y)}|$, where $|{\unb(x)}| = |{\unb(y)}|$ implies	$\unb(x) =\unb(y)$.
	\end{enumerate}

\begin{proposition} If \(\lambda_+\) and \(\lambda_-\) have no zero, 
	then every pair of compatible symbols \(a_i, \bar{a}_i \in \Sigma_n\) can be reduced from \((s_1, \ldots, s_n)\).\label{lemma:no-zero-implies-twisting-or-interlacing}
\end{proposition}
\noindent \Cref{lem:decomposition} follows directly, as whenever \((s_1, \ldots, s_n)\) is irreducible, \(\lambda_+\) and \(\lambda_-\) have a zero.

To show \Cref{lemma:no-zero-implies-twisting-or-interlacing} assume that for either \(b \in \Bool\) there is no \(x \in \mO^m \setminus 0^m\) such that \(\lambda_b(x) = 0\). By \Cref{lemma:ky-fan}, for each \(b \in \Bool\), there is a positively alternating \(m\)-chain, i.e., there are \(x_1^b \preceq x_2^b \preceq \cdots \preceq x_m^b \) and \(1 \leq j_1^b < j_2^b < \cdots < j_m^b \leq m\) such that \(\lambda_b(\{x_1^b, \ldots, x_m^b\}) = \{(-1)^{k-1} j_k^b \mid k \in [m]\}\).
We derive the following properties of the chains and \((s_1, \ldots, s_n)\):
\begin{enumerate}[label=(\alph*), leftmargin=2.5em]
	\item \label{item:onto} The values \(j_1^b, \ldots, j_m^b\) are pairwise distinct and in \([m]\). Thus, \(j_k^b = k\) for each \(k \in [m]\).
	\item By the definition of \(\prec\), every strictly increasing chain in \(\mO^m \setminus 0^m\) has at most length \(m\). Specifically, \(|x_1^b|_+ + |x_1^b|_{-} = 1\), \(|x_2^b|_+ + |x_2^b|_{-} = 2\), $\ldots$, and \(|x_m^b|_+ + |x_m^b|_{-} = m\). Also, for each \(k \in [m-1]\), \(x^b_k\) and \(x^b_{k+1}\) differ in exactly one position \(p\), i.e., \(x^b_k(p) = 0\) and \(x^b_{k+1}(p) \in \Bool\) and for each \(p' \neq p\): \(x^b_k(p') = x^b_{k+1}(p')\).
	\item \label{item:signedness:1}
	By \ref{item:onto}, all cases (\ref{case:hg2n}, \ref{case:singlesigned}, \ref{case:unbal}) of \(\lambda_b\) apply. 
	Because \(h\) is antitone, \ref{case:hg2n} applies to \(x_1^b, \ldots, x_{m-n-1}^b\) where 
	\(h(x^b_{k}) > h(x^b_{k+1})\) for each \(k \in [m-n-1]\).
	On the other hand, for \(k \in \{0, 1, \ldots, n\}\), we have that \(h(x^b_{m-n+k}) < 2n\). 
	
	Case \ref{case:singlesigned} applies to exactly one of the words \(x^b_{m-n}, \ldots, x^b_{m}\)  while to the remaining words case \ref{case:unbal} applies.
	For \ref{case:singlesigned} to apply, only one sign $\top_b \in \Bool$ may occur in the word, while \ref{case:unbal} requires that both signs occur. 
	Hence \ref{case:singlesigned} applies to \(x^b_{m-n}\) as the words, to which \ref{case:unbal} apply, need to be larger with respect to \(\preceq\).
	Moreover, as \(x^b_{m-n}\) is larger than \(x_1^b, \ldots, x^b_{m-n-1}\), the latter are in \(\{\top_b,0\}^m\), too.
	\(\top_b\) is determined by the sign of \(\lambda_b(x^b_{m-n}) = (-1)^n \cdot (n+1)\), i.e., it depends on \(b\) and the parity of \(n\). Precisely: \(\top_b = b \cdot (-1)^n\).
	
%
%
	\item \label{item:signedness:2}
	\ref{case:unbal} applies to \(x^b_{m-n+k}\) for each \(k \in [n]\) and yields values \(\{1, -2, 3,  \ldots, (-1)^{n-1}n\}\).
	As \(\unb(x^b_{m-n+k}) \neq 0\), we have that \({|{\unb(x^b_{m-n+k})}|} > {|{\unb(x^b_{m-n+k+1})}|}\) for all \(k \in [n-1]\). Thus, 
	\(\lambda(x^b_{m-n+k}) = (-1)^{n-k}\cdot(n-k+1)\). Denote \(n-k+1\) by \(i\).
	Observe that \(\unb(x^b_{m-n}) = 0\) because the word \({\top_b}^m \succeq x^b_{m-n}\) trivially balances every symbol.
	Hence, to ultimately unbalance $a_{i}$, a position \(p^b_{i}\) is newly signed in \(x^b_{m-n+k}\) with \(-\top_b\)
	where \(s(p^b_{i})\) is either \(a_{i}\) or \(\bar{a}_{i}\). 
	The exact choice of $p^b_{i}$ depends on \(\top_b\) and the parity of \(i\):
	If \((-1)^{i-1} = -\top_b\), then \(p^b_i\) is such that \(s(p^b_{i}) = a_{i}\), 
	otherwise, \(p^b_i\) is such that \(s(p^b_{i}) = \bar{a}_{i}\).
\end{enumerate}

Let \(U = \{p \in [m] \mid \exists b \in \Bool\colon x^b_{m-n}(p) = 0\} = \{p^b_i \mid b \in \Bool, i \in [n]\}\). 
From
\[\top_+= (+1) \cdot (-1)^n = (-1) \cdot (-1) \cdot (-1)^n = (-1) \cdot \top_- \] and \ref{item:signedness:2} it follows that
\(p^+_i \neq p^-_i\) and \(a_i, \bar{a}_i \in \{s(p) \mid p \in U\}\) for each \(i \in [n]\).
Hence, \(|U| = 2n\) and \(\Sigma_n \subseteq \{s(p) \mid p \in U\}\). To complete our proof, we further characterize the set \(U\).
\begin{property}\label{lemma:at-most-one-position-signed-in-boundary-cluster}
	Let \(p\) and \(p+1\) be neighboring endpoints. At most one of \(p\) and \(p+1\) is signed in \(x^b_{m-n}\).
\end{property}
\begin{proof}
	Assume the contrary. Let \(k \in [m-n]\) be minimal such that \(x^b_{k}(p) = \top_b = x^b_{k}(p+1)\), i.e., either \(x^b_{k-1}(p) = 0\) or \(x^b_{k-1}(p+1) = 0\). W.l.o.g.\@ assume that \(x^b_{k-1}(p+1) = 0\). Hence, using our algorithm to construct an element of \(H(x^b_{k-1})\), we may set \(p+1\) to \(\top_b\), i.e., we obtain \(x^b_k\). But then \(h(x^b_k) = h(x^b_{k-1})\), contradicting \ref{item:signedness:1}.
\end{proof}

 Let \(p\) and \(p+1\) be neighboring endpoints. We denote \(C(p) = C(p+1) = \{p, p+1\}\). For each position \(p' \in [m]\) that is not an internal endpoint, we denote \(C(p') = \{p'\}\). 
\begin{property}\label{lemma:unsigned-boundary-cluster-not-embraced-by-signs}
	Let \(1 < p < m\) be an unsigned position in \(x^b_{m-n}\) that is not an internal endpoint.
	Then \(C(p-1)\) or \(C(p+1)\) contain only unsigned positions in \(x^b_{m-n}\).
\end{property}
\begin{proof}
	Assume the contrary. Let \(k \in [m-n]\) be minimal such that both \(C(p-1)\) and \( C(p+1)\) contain a signed position in \(x^b_{k}\).
	W.l.o.g.\@ assume that \(C(p-1)\) contains a signed position already in \(x^b_{k-1}\). Using the algorithm to construct an element of \(H(x^b_{k-1})\), we may first set the other position in \(C(p-1)\) to \(\top_b\) (if it exists), then set \(p\) to \(-\top_b\), then set each position of \(C(p+1)\) to \(\top_b\) and call the intermediate result \(y\). We may also obtain \(y\) from \(x^{b}_{k}\) by first setting each unsigned position of \(C(p+1)\) to \(\top_b\), then \(p\) to \(-\top_b\), and then the remaining position of \(C(p-1)\) to \(\top_b\), if it exists.
	Thus, \(h(x^b_{k}) = h(x^b_{k-1})\),  contradicting \ref{item:signedness:1}.
\end{proof}	

	There are \(n-1\) pairs of neighboring endpoints, i.e., in total \(2n-2\) internal endpoints. \(x^+_{m-n}\) and \(x^-_{m-n}\) each have \(n\) unsigned positions, of which,
	by \Cref{lemma:at-most-one-position-signed-in-boundary-cluster}, at least \(n-1\) are internal endpoints. 
	\Cref{lemma:unsigned-boundary-cluster-not-embraced-by-signs} implies that the only unsigned position that is not an internal endpoint, if one exists, is \(1\) or \(m\). Thus, \(U\) consists solely of endpoints of the \(s_j\)'s.
	However, all symbols of \(\Sigma_n\) are distributed at the positions in \(U\). 
	Hence, \((s_1, \ldots, s_n)\) is such that every pair of compatible letters of \(\Sigma_n\) can be reduced. 
	This concludes the proof of \Cref{lemma:no-zero-implies-twisting-or-interlacing}.

\bibliography{biblio} \bibliographystyle{alpha}

\end{document}